\documentclass[10pt,twocolumn,twoside]{IEEEtran}
\usepackage{booktabs}
\usepackage{subfigure}
\usepackage{cite}
\usepackage{amsmath,epsfig,bm}
\renewcommand{\H}{\mathbf{H}}
\newcommand{\Hnull}{\mathcal{H}_0}
\newcommand{\Halt}{\mathcal{H}_1}

\newcommand{\Hin}{{X}}
\newcommand{\Hout}{{Y}}

\newcommand{\Honull}{\mathcal{{D}}_0}
\newcommand{\Hoalt}{\mathcal{{D}}_1}

\usepackage{pifont}
\usepackage{graphicx}
\usepackage{amssymb}
\usepackage{amsmath}
\usepackage{cancel}
\usepackage[normalem]{ulem}
\usepackage[dvips]{color}
\usepackage{algorithm}  
\newfont{\fsc}{eusm10}                         

\usepackage{lineno}
\modulolinenumbers[1]

\usepackage{algorithmic}
\usepackage{algorithm}
\newcommand{\INPUT}{\item[\algorithmicinput]}
\newcommand{\algorithmicinput}{\textbf{Input:}}
\newcommand{\OUTPUT}{\item[\algorithmicoutput]}
\newcommand{\algorithmicoutput}{\textbf{Output:}}

\newtheorem{theorem}{\noindent \textit{Theorem}}
\newtheorem{lemma}{\noindent \textit{Lemma}}

\title{Optimal Information-Theoretic Wireless Location Verification}
\author{
\IEEEauthorblockN{Shihao Yan$^1$,
                  Robert Malaney$^{1}$,
                  Ido Nevat$^2$,
                  and Gareth W. Peters$^{3}$\\
\IEEEauthorblockA{$^1$
School of Electrical Engineering \& Telecommunications,UNSW, Sydney, Australia.}\\
\IEEEauthorblockA{$^2$
Wireless \& Networking Tech. Lab, CSIRO, Sydney, Australia.}\\
\IEEEauthorblockA{$^3$
Department of Statistical Science, University College London, London, United Kingdom.}
} }

\begin{document}


\maketitle

\begin{abstract}

We develop a new Location Verification System (LVS) focussed on
network-based Intelligent Transport Systems and vehicular ad hoc networks. The algorithm we develop is based on an information-theoretic
framework which uses the received signal strength (RSS) from a network of base-stations and
the claimed position. Based on this information we derive the optimal decision regarding the verification
of the user's location. Our algorithm is optimal in the sense of maximizing the mutual information between its input
and output data. Our approach is based on the practical scenario in which a non-colluding malicious user
some distance from  a highway optimally boosts his transmit power in an
attempt to fool the LVS that he is on the highway. We develop a practical threat model for
this attack scenario, and investigate in detail the performance of the LVS in terms of its input/output
mutual information. We show how our LVS decision rule can be implemented straightforwardly with
a performance that delivers near-optimality under realistic threat conditions, with information-theoretic
optimality approached as the malicious user moves further from the highway. The practical advantages our new information-theoretic scheme delivers relative to more traditional Bayesian verification frameworks are discussed.


\end{abstract}

\section{Introduction}

The almost ubiquitous use of position information in emerging wireless networks has made the issue of  wireless location determination and  location-based services a very active research topic in recent years, \emph{e.g.} see (\cite{malaney2007nuisance,liu2008wireless,kuiper2011geographical,anisetti2011map,
tsalolikhin2011single,mourad2011robust,chang2012footprint,rabayah2012new,yang2013detecting,bialer2013maximum}). This in turn has made the supplementary issue of
Location verification in wireless networks an area of increasing importance. This is in part a consequence of not only the growing number of mobile services that utilize location information, but also in part due to the mission-critical nature the location information being supplied has on the performance, security  and safety of some services. The importance of location verification is perhaps best illustrated in emerging Intelligent Transport Systems (ITS) and vehicular ad hoc networks (VANETs)  where the verification of the location information supplied by vehicles is vital to the safety issues ITS (and VANETs) hope to address \cite{IEEE1609.2}. Indeed, recently there has been much effort in analyzing how a Location Verification System (LVS) in the context of ITS may operate \cite{sastry2003secure,leinmuller2006position, xiao2006detection, song2008secure, yan2009providing, ren2009location, yan2010cross, abumansoor2012secure}. Such ITS-based LVS work is also considered in other recent research efforts on location verification in more generic wireless network settings (\emph{e.g.} \cite{malaney2004location, vora2006secure, malaney2006secure, malaney2007securing, malaney2007wireless, capkun2008secure, zhang2008evaluation, chen2010detecting, liu2010node, chiang2012secure, yang2012securing, wei2012lightweight}).

 In an LVS one aims at verifying a user's claimed position based on some input measurements so as to perform a binary decision on whether the user is \emph{legitimate} (claims his true position) or {\emph{malicious}} (spoofs his claimed position).  In general, an LVS aims at obtaining a low false positive rate for {legitimate} users and a high detection rate for {malicious} users, leading to a tradeoff perhaps best illustrated by a receiver operating characteristic (ROC) curve. However, it is established that ROCs are not always ideal in comparing performances of two  separate systems (\emph{e.g.} \cite{gu2005information,gu2006measuring}). It is also the case that the use of a ROC does not in any formal sense indicate what the  \emph{optimal} operating point of an LVS is.  A possible direction to follow in attempting an optimization of an LVS is to utilize a Bayesian hypothesis test, which with uninformative priors contains the structure of a Likelihood Ratio Test (LRT) - which  minimizes the input/output classification error in the scenario where the cost of all types of misclassifications are equal \cite{chen2010detecting}. Additionally, if the costs of misclassifications are not equal, then a variation of the LRT  decision rule can be formed, namely the Bayes criterion \cite{barkat2005signal}. The Maximum A Posteriori criterion and the Maximum Likelihood criterion are special cases of the Bayes criterion. However, it is well known that these Bayes-decision criteria possess a weakness - they are \emph{subjective}. This subjectivity arises through the necessity to pre-assign costs to the different types of misclassification. It has been discussed before how such subjectivity in Bayes criteria can give rise to confusion when comparisons of detector performances are made \cite{gu2005information,gu2006measuring} . As such, although many of the previous works on LVSs have their own specific verification performance goals in mind, and their own pros and cons, none of these  works identify an optimal LVS in any non-subjective sense.

 To make progress, what is actually required is an \emph{objective} measure of detector performance, namely a single unified metric that takes into account all key aspects of intrusion detection in an objective fashion. As argued in \cite{gu2006measuring}, this metric should be the information-theoretic \emph{mutual information}, and it is this approach we develop here in the context of location verification.  More specifically, we develop  here  an information-theoretic framework for an LVS in which the mutual information between the input and output LVS data   is used as the objective optimization criterion. Some preliminary work along these lines has been attempted but only for sub-optimal decision rules\cite{yan2012information}. In this work we pursue an information-theoretic framework in which the decision rule is an optimal one.

In general an LVS can be characterized as follows. The input data (users to be verified) are represented by binary random variables $X = x$, $x \in [0,1]$,  whose realized elements indicate {legitimate} ($x = 0$) or {malicious} ($x = 1$). Likewise the output data can be represented by binary random variables $Y = y$, $y \in [0,1]$, whose realized elements indicate the binary decision made by the LVS, namely {\emph{verified}} ($y = 0$) or {\emph{not verified}} ($y = 1$). In the LVS, a \emph{decision rule} is formed which indicates whether a user is {malicious} or not. This decision rule ultimately forms a test on whether some statistic (derived from network measurements and some prior information) is less than or equal to some \emph{threshold}. With these definitions in place,  the contributions of this paper can be specifically summarized thus.

\begin{enumerate}
\item We develop for the first time an information-theoretic framework for an LVS, which allows us to utilize the mutual information between ${X}$ and ${Y}$ as a unique criterion to evaluate and optimize the performance of an LVS.
\item Under the assumption of \emph{known} likelihood functions for the measurements,  we prove that the likelihood ratio is the test statistic that produces the maximum mutual information between ${X}$ and ${Y}$.
\item Identifying the threshold value that maximizes the mutual information between ${X}$ and ${Y}$,  we then show how the Likelihood Ratio Test (LRT) is the decision rule which maximizes the mutual information between $X$ and $Y$, and leads to the information-theoretic optimal LVS. We take the further step of determining the likelihood functions under a series of threat models. This leads to a working LVS that will be an optimal information-theoretic approach under the given threat models.
\item We show from our analysis how an effectively optimal LVS, which is simple to deploy in practice, can be developed. We show that our LVS  leads to an optimal solution for most realistic attack scenarios in which a malicious user who is outside a network region, is attempting to spoof that he is within the network region. We further show how optimality is approached as the malicious user moves further from the network region.
\end{enumerate}

The remainder of the paper is structured as follows: Section \ref{Information Framework} presents both the general network system model and our information-theoretic LVS framework. The decision rule that optimizes mutual information is constructed in Section~\ref{LRT Framework}. In Section~\ref{RSS Scenario}, analysis and simulations of our LVS are  presented for a realistic threat model. Section~\ref{Conclusions} concludes the paper.

\section{System Model and LVS Framework}
\label{Information Framework}
In this section, we first present the general location verification system model and the related assumptions. Then, we develop an information-theoretic framework for an LVS, which allows us to utilize the mutual information between ${X}$ and ${Y}$ as an unique criterion to evaluate and optimize an LVS.

\subsection{General System Model} \label{System_Model}
The values of the input data $X$ can be represented as two hypotheses. The first of these is the null hypothesis, $\Hnull$, which assumes the user to be verified is legitimate ($x = 0$). The second one is the alternative hypothesis, $\Halt$, which assumes the user to be verified is malicious ($x=1$). Likewise, the possible values of the output data $Y$ can be represented as two decisions, where $\Honull$ denotes verified ($y = 0$), and $\Hoalt$ denotes not verified ($y = 1$). We now outline the general LVS model, and detail the assumptions we use.

\begin{enumerate}
\item {A \emph{single} user ({legitimate} or {malicious}) reports his claimed location, $\bm{\theta}^c = (u^c, v^c) \in \mathbb{R}^2$, to a network with $K$ ($K > 2$) Base Stations (BSs) in the communication range of the user (the $K$ BSs are not in a line), where $\bm{\theta}_{i}^{BS} = (u^{BS}_i, v^{BS}_i) \in \mathbb{R}^2$ is the location of the $i$th BS, $i = 1,2,...,K$. One of the $K$ BSs is the Process Center (PC), and all other BSs will transmit the measurements collected from the user to the PC. The PC is to make decisions based on the user's claimed location and the measurements collected by all the $K$ BSs.} We assume all BSs are perfectly synchronized.

\item We assume a user ({legitimate} or {malicious}) knows the locations of the $K$ BSs, and that $\bm{\theta}^c$ is supplied by the user to the PC.

\item For the legitimate user, we assume the true location is given by $\bm{\theta} = \bm{\theta}^c$ (here we will ignore the small location determination error, $e.g.$ the GPS error\footnote{When this error is much smaller than the average distance between BSs the effect on the results is negligible.}). We assume the malicious user's true location $\bm{\theta} = (u, v) \in \mathbb{R}^2$ is known exactly to him (\emph{i.e.} again we ignore any small localization error), but is unknown to the network.

\item We assume $\bm{\theta}$ is a bivariate random variable following some distribution. The prior distribution, $i.e.$ the probability density function (pdf), for $\bm{\theta}$ under $\Halt$ is denoted as $p(\bm{\theta}|\Halt)$.

 \item In general, the measurement ($M_i$) collected by the $i$th BS from a legitimate user is dependent on $\bm{\theta}_{i}^{BS}$ and the legitimate user's $\bm{\theta}^c$. In practice, a malicious user can impact the measurements collected by all BSs in order to avoid detection. Thus, the measurement ($M_i$) collected by the $i$th BS from a malicious user is some function of $\bm{\theta}_{i}^{BS}$ and the {malicious} user's $\bm{\theta}$ and his spoofed $\bm{\theta}^c$. Therefore, the measurement ($M_i$) collected by the $i$th BS can be given as a composite model as follows:
\begin{align}
\label{composite_model}
\begin {cases}
\begin{split}
    &\Hnull: M_{i} = h_0 \left(\bm{\theta}_{i}^{BS}, \bm{\theta}^c, \omega\right)\\
    &\Halt: M_{i} = h_1 \left(\bm{\theta}_{i}^{BS}, \bm{\theta}^c, \bm{\theta}, \omega\right),
\end{split}
\end {cases}
\end{align}
where
$h_0$ and $h_1$ are some functions yet to be specified (can involve additional parameters), and $\omega$ is random variable representing the communication channel noise. Given the statistical nature of $\omega$, the composite system model in (\ref{composite_model}) can produce the \emph{likelihood functions} under $\Hnull$ and $\Halt$, which are denoted as $p(\bm{m}|\Hnull)$ and $p(\bm{m}|\Halt)$, respectively,  where $\bm{m} = [m_1, m_2, ..., m_K]$ is a realization of the measurement vector, $\bm{M} = [M_1, M_2, ..., M_K]$.

\item {We also assume a user is legitimate with a known prior probability, which is $P_0 = P(x = 0)$. The probability of a user to be malicious is denoted as $P_1$, and $P_1 = 1-P_0$.}
\end{enumerate}

\begin{figure}[t]
    \centering
        \epsffile{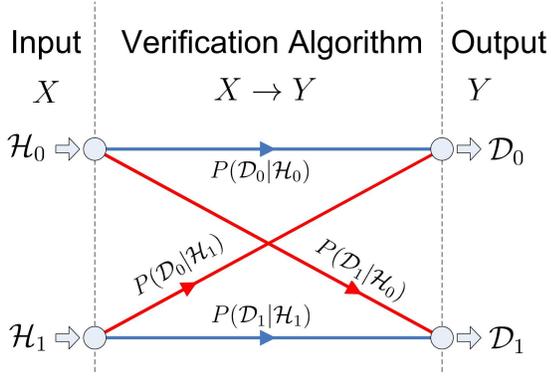}
        \caption{A Location Verification System (LVS) model.}
    \label{fig:channel_model}
\end{figure}

\subsection{Information-Theoretic Framework for an LVS}

In general, the purpose of an LVS is to map the input data $\Hin$ to the output data $\Hout$, $X \rightarrow Y$, and can be represented as shown in Fig.~\ref{fig:channel_model}. In this figure, the false positive rate, $\alpha$, and the detection rate, $\beta$, are given as follows
\begin{align*}
&\alpha = P(\Hoalt|\Hnull),~~1 - \alpha = P(\Honull|\Hnull),\\
&\beta = P(\Hoalt|\Halt),~~1 - \beta = P(\Honull|\Halt),
\end{align*}
where $P$ is the probability of an outcome conditional on a hypotheses. The mutual information between $\Hin$ and $\Hout$ can be expressed as $I({\Hin};{\Hout}) = H({\Hin})-H({\Hin}|{\Hout})$, where $H({\Hin})$ is the entropy of $\Hin$, and $H({\Hin}|{\Hout})$ is the conditional entropy of $\Hin$ given $\Hout$. Given $P_0$, the entropy of the discrete binary random variable $X$ can be written as $H(\Hin) = -\sum_{x}P(X)\log_2 P(X) = -P_0\log_2 P_0 - (1-P_0)\log_2(1 - P_0)$.
 With these definitions, the conditional entropy $H({\Hin}|{\Hout})$ can be expressed as in \cite{gu2005information}
\begin{equation}
\begin{split}
H({\Hin}|&{\Hout}) = -\sum_{x}\sum_{y}P(X,Y)\log_2 P(X|Y)\\
=~ &P_0(1\!-\!\alpha) \left(\!-\!\log_2\frac{P_0(1\!-\!\alpha)}{P_0(1\!-\!\alpha)+(1\!-\!P_0) (1\!-\!\beta)}\right)\\
 & + P_0\alpha\left(\!-\! \log_2\frac{P_0\alpha}{P_0\alpha \!+\! (1\!-\!P_0)\beta}\right)\\
& + (1\!-\!P_0)(1\!-\!\beta) \left(\!-\! \log_2\frac{(1\!-\!P_0)(1\!-\!\beta)}{P_0(1\!-\!\alpha)\!+\! (1\!-\!P_0)(1\!-\!\beta)}\right)\\
& + (1\!-\!P_0)\beta \left(\!-\!\log_2\frac{(1\!-\!P_0)\beta}{P_0\alpha\!+\! (1\!-\!P_0) \beta}\right).
\end{split}
\label{conentropy}
\end{equation}

{The mutual information $I(X, Y)$ measures the reduction of uncertainty of the input $X$ given the output $Y$. For example, if we make verification decisions without any observations (\emph{e.g.} of received signal strengths), $X$ and $Y$ will be independent of each other, and $I({X}; {Y})$ will be minimized (zero). However, based on some observations our LVS attempts to map $X$ into $Y$ so as to minimize the uncertainty of $X$ given $Y$. An extreme example of this is when $X$ and $Y$ are identical and therefore $I({X}; {Y})$ is maximized (of course this would require infinite noisy observations or finite noiseless observations). More generally, given some finite noisy observations, maximizing on the mutual information $I({X}; {Y})$ leads to decisions which maximize the dependence of $X$ and $Y$. As such,  the mutual information  is the natural optimization metric for an LVS from an information-theoretic viewpoint.} The information-theoretic optimal location verification algorithm can be defined as the one which maximizes $I({\Hin};{\Hout})$ defined above.

\section{The Optimal Location Verification Algorithm}
\label{LRT Framework}

In this section, based on the assumption of known likelihood functions under both $\Hnull$ and $\Halt$, we take the additional step of identifying the information-theoretic optimal location verification algorithm, which produces the maximum $I({\Hin};{\Hout})$ relative to any other location verification algorithm.

\subsection{The Decision Rule for Maximizing $I({\Hin};{\Hout})$}

In the context of an LVS, a location verification algorithm must formulate a decision rule to infer whether the user is consistent with $\Hnull$ or $\Halt$. The algorithm ultimately forms a comparison of some test statistic, $F(\bm{m})$, and a corresponding threshold, $T_F$, in the form of
\begin{align*}
F(\bm{m}) \begin{array}{c}
\overset{\Hoalt}{\geq} \\
\underset{\Honull}{<}
\end{array}%
T_F.
\end{align*}
For a given $F(\bm{m})$, we will be interested in the value of $T_F$ which maximizes $I({\Hin};{\Hout})$, $i.e.$, $T_F^* = \arg\underset{T_F}{\max}\,{I({\Hin};{\Hout})}$. Furthermore, we will be interested in determining the functional form of $F(\bm{m})$ that maximizes $I({\Hin};{\Hout})$. This leads to our main result, which is stated in Theorem~1.
\begin{theorem}
\label{theorem 1}
{Given the decision rule
\begin{equation}
\label{arbitrary}
F(\bm{m}) \begin{array}{c}
\overset{\Hoalt}{\geq} \\
\underset{\Honull}{<}
\end{array}%
T_F^*,
\end{equation}
the functional form of $F(\bm{m})$ that maximizes the mutual information $I({\Hin};{\Hout})$ is $\Lambda\left(\bm{m}\right)$, where
\begin{equation}\label{likelihood_ratio}
\Lambda\left(\bm{m}\right) \triangleq \frac{p\left(\bm{m}|\Halt\right)}{p\left(\bm{m}|\Hnull\right)}.
\end{equation}
}
\end{theorem}

To prove Theorem~\ref{theorem 1}, we first introduce two lemmas, of which the first is the Neyman-Pearson Lemma \cite{neyman1933problem}.
\begin{lemma}
\label{lemma 1}
{
Consider two hypotheses $\Hnull$ and $\Halt$, the decision rule to maximize a detection rate ($\beta$) for a given false positive rate ($\alpha$) is
\begin{align}\label{Neyman}
\Lambda\left(\bm{m}\right)
\begin{array}{c}
\overset{\Hoalt}{\geq} \\
\underset{\Honull}{<}
\end{array}%
T_{\Lambda},
\end{align}
where $T_{\Lambda}$ is determined by the specified value of $\alpha$. For proof, see \cite{neyman1933problem}. Before proceeding, we note $\alpha < \beta$ will be a basic requirement for any useful LVS.}
\end{lemma}

\begin{lemma}
\label{lemma 2}
{
 Given the assumption $\alpha < \beta$, the mutual information $I({\Hin};{\Hout})$ is a monotonic increasing function of the detection rate $\beta$.}
\end{lemma}
\begin{proof}[Proof of Lemma 2]
Since $H(\Hin)$ is not dependent on $\beta$, the first derivative of $I({\Hin};{\Hout})$ with respect to $\beta$ can be expressed as
\begin{equation*}
\begin{split}
\frac{\partial I({\Hin};{\Hout})}{\partial \beta} &= \frac{\partial \left(H(\Hin) - H({\Hin}|{\Hout})\right)}{\partial \beta} = - \frac{\partial H({\Hin}|{\Hout})}{\partial \beta}\\
=~& (1-P_0) \left(\log_2\frac
{\beta[P_0(1-\alpha)+(1-P_0)(1-\beta)]}{(1-\beta)[P_0\alpha+(1-P_0)\beta]}\right)\\
 =~&(1-P_0) \log_2 \underbrace{\{\beta [P_0(1-\alpha)+(1-P_0)(1-\beta)]\}}_{V_{0}}\\
& - (1-P_0)\log_2 \underbrace{\{[P_0\alpha+(1-P_0)\beta](1-\beta)\}}_{V_{1}}.
\end{split}
\end{equation*}
Note, since $0 < P_0 < 1$, and the logarithm is a monotonic increasing function of $z$, then $\partial I({\Hin};{\Hout})/{\partial \beta}$ has the same sign as $(V_{0}~-~V_{1})$, where
\begin{align*}
V_{0} - V_{1} =~& \beta [P_0(1-\alpha)+(1-P_0)(1-\beta)]\\
 & - [P_0\alpha+(1-P_0)\beta](1-\beta)\\
=~& P_0 (\beta - \alpha).
\end{align*}
Thus, given the assumption $\alpha < \beta$, then $\partial I({\Hin};{\Hout})/{\partial \beta} > 0$, and Lemma \ref{lemma 2} is proved.
\end{proof}

Given Lemma~\ref{lemma 1} and Lemma~\ref{lemma 2}, we now prove Theorem~\ref{theorem 1}.\\
\begin{proof}[Proof of Theorem 1]
If the specified value of $\alpha$ in Lemma~\ref{lemma 1} is that which results in the value ${T_F^*}$ of (\ref{arbitrary}), then by Lemma~\ref{lemma 2} the result follows.
\end{proof}

\subsection{The Optimal Location Verification Algorithm}

Based on the above discussion, the optimal information-theoretic location verification algorithm is presented in Algorithm \ref{alg:sfr_el}.
\algsetup{indent=2em}
\begin{algorithm}[!hb]
\caption{Optimal Location Verification Algorithm}
\label{alg:sfr_el}
\begin{algorithmic}[1]
\INPUT priori probability $P_0$, likelihood functions.
\OUTPUT binary decisions $\Hoalt$ and $\Honull$.
\medskip
\STATE Determine the functional forms of $h_0$ and $h_1$ in (\ref{composite_model}).
\STATE Specify the prior distributions for $\bm{\theta}$, $p(\bm{\theta}|\Hnull)$ and $p(\bm{\theta}|\Halt)$, and determine the likelihood functions $p(\bm{m}|\Hnull)$ and $p(\bm{m}|\Halt)$.
\STATE With (\ref{Neyman}) as the general decision rule, derive the functional form of $\alpha$ and $\beta$. Note, $\alpha$ and $\beta$ will be functions of $T_{\Lambda}$.
\STATE Using $I(X;Y)$ as the objective function, search for $T^*_{\Lambda}$, which is the value of $T_{\Lambda}$ that maximizes $I(X;Y)$.
\STATE Collect measurements and calculate the likelihood ratio $\Lambda(\bm{m})$ according to the likelihood functions determined in step~2.
\STATE Form the optimal decision rule,
\begin{align}
\label{optimal decision rule}
\Lambda(\bm{m})
\begin{array}{c}
\overset{\Hoalt}{\geq} \\
\underset{\Honull}{<}
\end{array}%
T_{\Lambda}^*.
\end{align}
\end{algorithmic}
\end{algorithm}
%
%
%
%

\section{Specific Optimal Location Verification Algorithm with RSS as Measurements}
\label{RSS Scenario}

In order to implement the optimal location verification algorithm, in this section we take the further step of determining the likelihood functions under $\Hnull$ and $\Halt$ with Received Signal Strength (RSS) as the system measurements, and we consider the algorithm under a series of threat models.

Although the framework we develop can be built on any measurement (location information metric), such as RSS, TOA (time of arrival) and TDOA (time difference of arrival), for purposes of illustration we focus here only on an RSS implementation. In this case, the measurement $M_i$ is the RSS (in dB) collected by the $i$th BS. We will also assume that the legitimate user and all BSs are equipped with only a single omni-direction antenna.

 Let us define the set of BSs that are within range  of a legitimate user positioned at $\bm{\theta}^c$ as the  \emph{in-range} BSs. This set of BSs forms an effective perimeter for the network used in the location verification.  We will assume a single malicious user has the technology (\emph{e.g.}  directional  beam-forming), which allows him to ensure (if required) that from some position outside the perimeter only the in-range BSs  receive a non-zero RSS. The malicious user can set the power of the main directional beam. We do \emph{not} allow an adversary to set multiple  beams to different BSs via  colluding malicious users (see later discussion).

 Based on the log-normal propagation model \cite{goldsmith2005wireless}, $h_0$ in (\ref{composite_model}) can be specified as
\begin{align}\label{h0_RSS}
h_0 (\bm{\theta}_{i}^{BS}, \bm{\theta}^c, \omega) = P_0 - 10 \gamma  \log_{10}\left(\frac{d_i^c}{d_0}\right) + \omega,
\end{align}
where $P_0$ is a reference received power, $d_0$ is the reference distance, $\gamma$ is the path loss exponent, $\omega$ (in dB) is a zero-mean normal random variable with variance $\sigma_{dB}^2$, and the Euclidean distance of the $i$th BS to the user's claimed location $\bm{\theta}^c = (u^c, v^c)$ is
\begin{align*}
d_i^c = d (\bm{\theta}^c,\bm{\theta}_{i}^{BS}) = \sqrt{(u^c - u^{BS}_i)^2 + (v^c - v^{BS}_i)^2}.
\end{align*}
A malicious user can adjust his transmit power to impact the measurements collected by the BSs, thus $h_1$ in  (\ref{composite_model}) can be expressed as
\begin{align}\label{h1_RSS}
h_1 (\bm{\theta}_{i}^{BS}, \bm{\theta}^c, \bm{\theta}, \omega) = P_0 + P_x - 10 \gamma  \log_{10}\left(\frac{d_i^t}{d_0}\right) + \omega,
\end{align}
where $P_x$ is a power level set by the malicious user, and $d_i^t = d (\bm{\theta},\bm{\theta}_{i}^{BS})$ is the Euclidean distance of the $i$th BS to the user's true location $\bm{\theta} = (u, v)$.

Assuming all  $M_i$'s are independent from each other, the likelihood function $p(\bm{m}|\Hnull)$ can be expressed as
\begin{equation}
\begin{split}
p(\bm{m}| \Hnull) &= \prod_{i=1}^{K}p({m_i}|\Hnull)\\
&= \prod_{i=1}^{K}\frac {1}{\sqrt{2\pi}\sigma_{dB}}\exp\left(\frac{-\left(m_i - \mu^c_i \right)^2}{2\sigma_{dB}^2}\right),
\end{split}
\label{likelihood0_one}
\end{equation}
where
\begin{align*}\label{mean}
\mu_i^c \!=\! P_0 \!-\! 10 \gamma  \log_{10}\left(\frac{d_i^c}{d_0}\right).
\end{align*}
Also, the pdf of $\bm{M}$ conditional on $\bm{\theta}$ under $\Halt$, $p(\bm{m}|\bm{\theta}, \Halt)$, can be written as
\begin{equation}
\begin{split}
&p(\bm{m}|\bm{\theta},\Halt) = \prod_{i=1}^{K}p({m_i}|\bm{\theta}, \Halt)\\
 &= \prod_{i=1}^{K}\frac {1}{\sqrt{2\pi}\sigma_{dB}}\exp\left(\frac{\!-\!\left(m_i\!-\!P_0\!-\! P_x \!+\! 10 \gamma  \log_{10}\left(\frac{d_i^t}{d_0}\right) \right)^2}{2\sigma_{dB}^2}\right).
\end{split}
\label{Likelihood1}
\end{equation}

In general, a malicious user will utilize $P_x$ in an attempt to impact the measurements collected by the BSs in order to avoid detection. We now discuss how to determine the `optimal' value of $P_x$ from a malicious user's point of view. An LVS can be spoofed optimally if the measurements collected from a malicious user follow exactly $p(\bm{m}|\Hnull)$, which is given by (\ref{likelihood0_one}). Therefore, in order to avoid detection a malicious user attempts to minimize the difference between $p(\bm{m}|\Hnull)$ and $p(\bm{m}|\bm{\theta}, \Halt)$. This difference can be quantified through the KL-divergence between the two likelihood functions, which is defined as follows \cite{cover2006elements}
\begin{equation*}
\begin{split}
D_{KL}(p(&\bm{m}|\Hnull)|| p(\bm{m}|\bm{\theta}, \Halt))\\
&= \int p(\bm{m}|\Hnull) \ln {\frac{p(\bm{m}|\Hnull)}{p(\bm{m}|\bm{\theta}, \Halt)}} d{\bm{m}}\\
&= \sum_{i=1}^{K} \frac{\left(\mu_i^c - \mu_i^t - P_x \right)^2}{2\sigma_{dB}^2},
\end{split}
\label{KLdistance}
\end{equation*}
where
\begin{align*}\label{mean}
\mu_i^t \!=\! P_0 \!-\! 10 \gamma  \log_{10}\left(\frac{d_i^t}{d_0}\right).
\end{align*}
This KL-divergence is the information lost when $p(\bm{m}|\bm{\theta}, \Halt)$ is used to approximate $p(\bm{m}|\Hnull)$, and it becomes zero \emph{if and only if} the two distributions are identical. From an information-theoretic point of view the optimal value of $P_x$ can be expressed as
\begin{equation}
\begin{split}
P_x^* &= \arg\underset{P_x}{\min}\,{D_{KL}\left(p(\bm{m}|\Hnull|| p(\bm{m}|\bm{\theta}, \Halt)\right)}\\
&= \frac{1}{K}\sum_{i=1}^{K} \left(\mu_i^c - \mu_i^t \right).
\end{split}
\label{optimal Px}
\end{equation}
Setting $P_x = P_x^*$ in (\ref{h1_RSS}), $h_1$ can be rewritten as
\begin{align}\label{h1_RSS_general}
h_1 (\bm{\theta}_{i}^{BS}, \bm{\theta}^c, \bm{\theta}, \omega) \!=\! P_0 \!+\! \bar{\mu}^c \!-\! \bar{\mu}^t \!-\! 10 \gamma \log_{10}\left(\frac{d_i^t}{d_0}\right) \!+\! \omega,
\end{align}
where
\begin{align*}\label{Pt and Pc}
\bar{\mu}^c = \frac{1}{K}\sum_{i=1}^K \mu^c_i, ~\mathrm{and}~~\bar{\mu}^t = \frac{1}{K}\sum_{i=1}^K \mu^t_i.
\end{align*}
Although $\bm{\theta}$ is a known deterministic parameter for a malicious user, it is unknown for the network. This means $h_1$ is still unknown, and therefore the likelihood function $p(\bm{m}|\Halt)$ is unknown for the LVS. To make progress, we will assume some realistic threat models within which $p(\bm{m}|\Halt)$ becomes known.

\subsection{Threat Model}

\begin{figure}[t]
    \centering
        \epsfysize=6.5cm
        \epsfxsize=7cm
        \epsffile{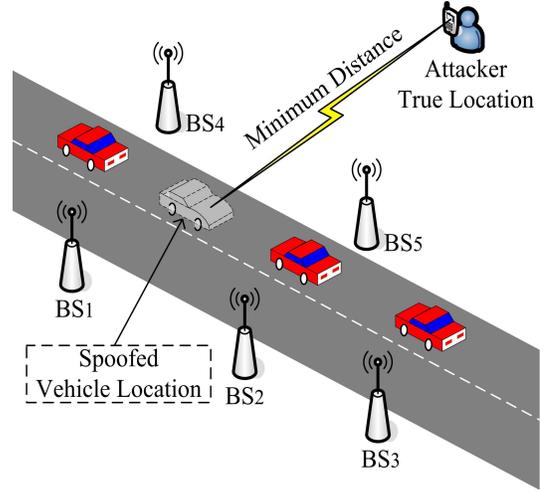}
        \caption{Illustration of the Minimum Distance (MD) threat model.}
    \label{fig:threat_model}
\end{figure}

The threat model we adopt can in principle accommodate any true/spoofed location pair. However, as the spoofed location approaches the true location, our detection  rate will approach zero (as expected for any verification system). As such, in the following we will make the assumption that the true position of an attacker is some minimum distance from the spoofed position. To quantify this we will assume that this distance is \emph{always greater than the mean separation distance between  BSs}. This is a reasonable assumption since it is unlikely an attacker will try to spoof a location too close to his actual locale, for fear of apprehension. Pragmatically, this assumption also means that in effect the attacker will be placed at some minimum distance \emph{off} the highway, since we find that if he is on the highway, under our minimum distance assumption, he is trivially detectable.\footnote{A  malicious user (vehicle) on the highway claiming to be at a position which is at least the mean separation distance between  BSs, must boost its transmissions significantly. However, since  the user is on the actual highway, the nearest BS to its true location will easily detect an attack. This is confirmed by our analysis, and as such we will henceforth consider the attacker to be sophisticated enough to realize that to have any reasonable chance of remaining undetected he must launch his attacks from a position at some minimum distance from the closet VANET BS (\emph{i.e.}  off the highway).}

  We henceforth refer to our generic threat model as the \emph{Minimum Distance} (MD) threat model, a schematic of which is given in Fig.~\ref{fig:threat_model}. In this scenario it is assumed the BSs that form the infrastructure part of the VANET, are placed alongside the highway (or on overhanging structures along the highway). This represents a realistic expectation for the physical deployment architecture for VANETs that is emerging from the ITS community.

However, before presenting the details of the DM threat model, we will first consider some simplifying approximations, which although not adopted in the DM threat model, do allow for additional insight and analytical clarity. We will also show how the optimal threshold derived for the MD threat model
is effectively the same as the optimal threshold derived under some of the simplifying approximations.

\subsubsection{Far Field Approximation (FFA)} \label{Simple Scenario}

In this subsection we propose the deployment of our LVS within a threat model where the far-field approximation (FFA) is made, meaning that the malicious user's distance from the highway is far enough that we can assume all RSS values received by all BSs are equal. Although never achieved in practice this simplification will allow us some initial insight into the performance of the LVS. Under the FFA  we can take the distance of a {malicious} user's true location to every BS to be approximated as a constant, $i.e.$ $d_i^t = d (\bm{\theta},\bm{\theta}_{i}^{BS}) \cong constant, \forall i \in [1, 2, \dots, K]$. Therefore, we will assume
\begin{equation}\label{Pt}
\frac{1}{K}\sum_{i=1}^K 10 \gamma \log_{10}\left(\frac{d_i^t}{d_0}\right) = 10 \gamma \log_{10}\left(\frac{d_i^t}{d_0}\right).
\end{equation}
Substituting  (\ref{Pt}) into (\ref{h1_RSS_general}), $h_1$ under the FFA can be expressed as
\begin{align}\label{IFA measurement model}
h_1 \left(\bm{\theta}_{i}^{BS}, \bm{\theta}^c, \omega\right) = \bar{\mu}^c + \omega.
\end{align}
Then, $p(\bm{m}|\Halt)$ (which now does not depend on $\bm{\theta}$) can be written as
\begin{align}\label{likelihood1_one}
p(\bm{m}|\Halt)  = \prod_{i=1}^{K}\frac {1}{\sqrt{2\pi}\sigma_{dB}}\exp\left(\frac{-\left(m_i - \bar{\mu}^c \right)^2}{2\sigma_{dB}^2}\right).
\end{align}
Based on (\ref{likelihood_ratio}), (\ref{likelihood0_one}) and (\ref{likelihood1_one}), we construct the decision rule
\begin{equation}
\label{IFA_decision_rule}
\Lambda\left(\bm{m}\right) = \frac{\exp\left(\frac{-\sum_{i=1}^{K}(m_i - \bar{\mu}^c)^2}{2\sigma_{dB}^2}\right)}{\exp\left(\frac{-\sum_{i=1}^{K}(m_i - \mu_i^c)^2}{2\sigma_{dB}^2}\right)}
\begin{array}{c}
\overset{\Hoalt}{\geq} \\
\underset{\Honull}{<}
\end{array}%
T_{\Lambda}.
\end{equation}
In order to help determine $\alpha$ and $\beta$ analytically, this decision rule can be rewritten as
\begin{subequations}\label{IFA_new_decision_rule}
\begin{align}
\digamma(\bm{m})
\begin{array}{c}
\overset{\Hoalt}{\geq} \\
\underset{\Honull}{<}
\end{array}
\Gamma,
\end{align}
where
\begin{align}
\digamma(\bm{m}) = {\mathop{\sum}\limits_{i=1}^{K}}m_i (\bar{\mu}^c - \mu_i^c),
\end{align}
and
\begin{align}\label{IFA_threshold}
\Gamma = \frac{1}{2}\left(2\sigma_{dB}^2 \ln {T_{\Lambda}} - {\mathop{\sum}\limits_{i=1}^{K}}\left((\mu_i^c)^2 -(\bar{\mu}^c)^2\right)\right).
\end{align}
\end{subequations}
Given (\ref{likelihood0_one}) and (\ref{IFA_new_decision_rule}) , we have
\begin{align*}
&p\left(\digamma(\bm{m})|\Hnull\right) = N\left(\sum_{i=1}^{K} \mu_i^c \left(\bar{\mu}^c - \mu_i^c\right), \sum_{i=1}^{K} \left(\bar{\mu}^c - \mu_i^c\right)^2 \sigma_{dB}^2\right),
\end{align*}
where $N(a, b)$ represents a normal distribution with $a$ and $b$ as the mean and variance, respectively. Likewise, given (\ref{likelihood1_one}) and (\ref{IFA_new_decision_rule}), we have
\begin{align*}
&p\left(\digamma(\bm{m})|\Halt \right) = N\left(\sum_{i=1}^{K} \bar{\mu}^c \left(\bar{\mu}^c - \mu_i^c\right), \sum_{i=1}^{K} \left(\bar{\mu}^c - \mu_i^c\right)^2 \sigma_{dB}^2\right).
\end{align*}
The false positive and detection rates under the FFA can now be expressed analytically as
\begin{equation}
\begin{split}
\alpha = P(&\Lambda(\bm{m}) > T_{\Lambda}|\Hnull) = P(\digamma(\bm{m}) > \Gamma|\Hnull)\\
&= {\mathcal{Q}}\left(\frac{\Gamma - \sum_{i=1}^{K} \mu_i^c \left(\bar{\mu}^c - \mu_i^c\right)}{\sqrt{\sum_{i=1}^{K} \left(\bar{\mu}^c - \mu_i^c\right)^2 \sigma_{dB}^2}}\right),
\end{split}
\label{IFA_alpha}
\end{equation}
\begin{equation}\label{IFA beta}
\begin{split}
\beta = P(&\Lambda(\bm{m}) > T_{\Lambda}|\Halt) = P(\digamma(\bm{m}) > \Gamma|\Halt)\\
&= {\mathcal{Q}}\left(\frac{\Gamma - \sum_{i=1}^{K} \bar{\mu}^c \left(\bar{\mu}^c - \mu_i^c\right)}{\sqrt{\sum_{i=1}^{K} \left(\bar{\mu}^c - \mu_i^c\right)^2 \sigma_{dB}^2}}\right),
\end{split}
\end{equation}
where ${\mathcal{Q}}(x) = \frac{1}{\sqrt{2 \pi}}\int_x^{\infty} \exp^{-t^2/2} dt$.

Having determined $\alpha$ and $\beta$ under the FFA, we can use these in (2) for the conditional entropy $H(X|Y)$. The value of $\Gamma$ which maximizes $I({\Hin};{\Hout}) = H({\Hin})-H({\Hin}|{\Hout})$, denoted as $\Gamma^*$, can be determined numerically. Using (\ref{IFA_threshold}), the $T_{\Lambda}^*$ can be determined by $\Gamma^*$. Then, the decision rule in (\ref{optimal decision rule}) which leads to the optimal verification algorithm under the FFA can be formed, where $\Lambda(\bm{m})$ is specified in (\ref{IFA_decision_rule}).

\begin{figure}[t]
    \centering
        \epsfysize=7cm
        \epsffile{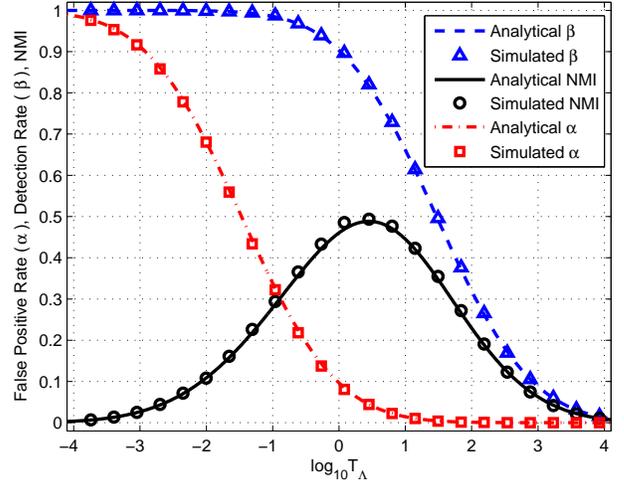}
        \caption{Analytical and simulated $\alpha$, $\beta$ and Normalized Mutual Information (NMI) with different values of $T_{\Lambda}$, where $K = 10, L = 1$, and $\sigma_{dB} = 5$.}
    \label{fig:IFA}
\end{figure}

We verify the false positive and detection rates, given by (\ref{IFA_alpha}) and (\ref{IFA beta}), respectively, via detailed Monte Carlo simulations. The simulation settings are chosen so as to mimic a location verification test over an area spanning the intersection of several major freeways:
\begin{itemize}
\item The $K$ BSs are randomly distributed in a $200m\times 200m$ square area, $4 \leq K \leq 10$.
\item The claimed locations of legitimate and malicious users are the same, which is $\bm{\theta}^c = [0, 0]$.
\item The legitimate users are at $\bm{\theta}^c$. The malicious users are infinitely far away from $\bm{\theta}^c$, which in practice means the measurements collected are generated according to (\ref{IFA measurement model}).
\item Each BS collects $L$ measurements from each user.
\end{itemize}

\begin{figure}[t]
    \centering
        \epsfysize=7cm
        \epsffile{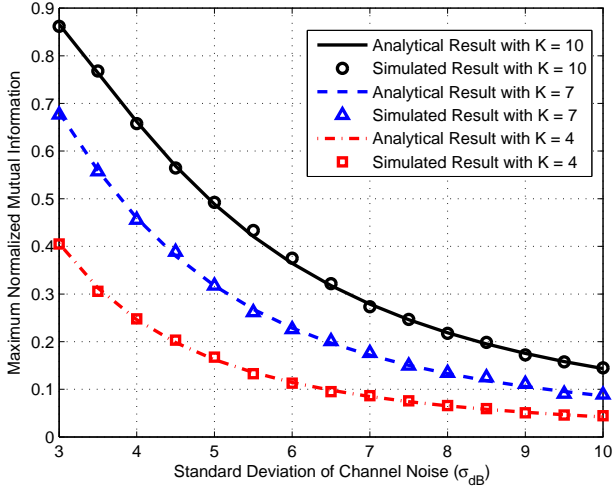}
        \caption{Maximum Normalized Mutual Information (NMI) with different values of $K$ and $\sigma_{dB}$, where $L = 1$.}
    \label{fig:IFA_sigma}
\end{figure}

In the following, the simulation results are obtained through 10,000 Monte Carlo realizations of the measurement vector $\bm{M}$, and in all the specific results shown we have adopted the prior probability $P_0 = 0.9$, and the path loss exponent $\gamma = 3$. Also note, we denote the Normalized Mutual Information (NMI) as
\begin{equation}\label{NMI}
\textrm{NMI} = {I(\Hin; \Hout)}/{H(\Hin)}.
\end{equation}

In Fig.~{\ref{fig:IFA}}, the analytical $\alpha$ and $\beta$ are directly derived from (\ref{IFA_alpha}) and (\ref{IFA beta}), respectively, and the analytical NMI is calculated using (\ref{conentropy}), (\ref{IFA_alpha}), (\ref{IFA beta}) and (\ref{NMI}).  In order to obtain the simulated $\alpha$, we randomly generate $\bm{M}$ according to (\ref{h0_RSS}), from which we get a specific realization of $\Lambda(\bm{m})$, and for each value of $T_{\Lambda}$ we decide whether the user is legitimate or malicious by (\ref{IFA_decision_rule}). To obtain the simulated $\beta$, we randomly generate $\bm{M}$ according to (\ref{IFA measurement model}), and follow the same procedure as above for $\alpha$. The simulated NMI are calculated using (\ref{conentropy}) and (\ref{NMI}) with the simulated $\alpha$ and $\beta$ as the input. In Fig.~{\ref{fig:IFA}} we have set $K = 10$, $L = 1$, and $\sigma_{dB} = 5$. From Fig.~{\ref{fig:IFA}}, we can see that the comparison between simulation and analysis shows excellent agreement, which verifies the analysis we have provided under the FFA. {As we can see from Fig.~\ref{fig:IFA}, relatively high false positive rates are found at low thresholds. This is a consequence of the LVS operating at a point far from the optimal threshold. However, we see that at the information-theoretic optimal threshold, which maximizes the NMI, the false positive rate is approximately 4\%.  This strong dependence on the threshold (also seen in all our other results),  re-emphasizes the critical importance of always operating the LVS at the optimal threshold.} We have investigated a range of other values of $K$, $L$ and $\sigma_{dB}$. Some of these results are shown in Fig.~\ref{fig:IFA_sigma}, where the maximum NMI is shown as a function of $\sigma_{dB}$ for different values of $K$. From Fig.~\ref{fig:IFA_sigma}, we see again the simulations agree with the analytical results.

\subsubsection{Uniformly Distributed Approximation (UDA)}\label{circle}

In this subsection we propose  the Uniformly Distributed Approximation (UDA), where the malicious users are assumed to be uniformly distributed on a circle. Again, although never achievable in practice this simplification will allow us additional insights. More specifically, the {malicious} user's true location is uniformly distributed on a circle, whose radius and center are $R$ and $\bm{\theta}^c$, respectively, where $R > r$, and $r = \max (d^c_i)$.

The main purpose of this model is to commence our probe of how reliable the use of the FFA will be when its assumptions are violated. To this end, we note that if the maximum difference between any two measurements collected from a malicious user, $M_i$ and $M_j$ $(i \neq j)$, is no larger than $\sigma_{dB}$, the scale $R$ at which this occurs provides a \emph{natural} distance at which we could anticipate the FFA and the UDA to be approximately equivalent. To quantify this let us introduce $ \rho \triangleq R/r$. Under the UDA, the difference between $M_i$ and $M_j$ can be written as
\begin{equation}
\begin{split}
\left| M_i \!-\! M_j \right| &= \left|10 \gamma \log_{10} \left(\frac{d_i^t}{d_0}\right) \!-\! 10 \gamma \log_{10} \left(\frac{d_j^t}{d_0}\right)\right|.
\end{split}
\label{differenceM}
\end{equation}
Given that for a malicious user, we have  $\max(|d_i^t - d_j^t|) = 2 r$, and $\min (d_i^t) = R - r$, we can write (\ref{differenceM}) as
\begin{align*}
\left| M_i \!-\! M_j \right| \leq \frac{10 \gamma}{\ln 10} \ln \left(\frac{d_i^t \!+\! 2 r}{d_i^t}\right) \leq
\frac{10 \gamma}{\ln 10} \ln \left(1 \!+\! \frac{2 r}{R \!-\! r}\right),
\end{align*}
where without loss of generality  we have assumed $d_j^t > d_i^t$.
In order to guarantee the required constraint $\max|M_i - M_j| \leq \sigma_{dB}$, we should have
\begin{align*}
\frac{10 \gamma}{\ln 10} \ln \left(1 + \frac{2 r}{R - r}\right) \leq \sigma_{dB},
\end{align*}
which results in
\begin{align}\label{Ratio Bound}
\rho \geq \rho^* \triangleq \frac{2}{\exp\left\{\frac{\sigma_{dB}\ln 10}{10 \gamma}\right\} -1} + 1,
\end{align}
where $\rho^*$ is a reference value that will be utilized when comparison under the FFA is made. Such a comparison is achieved by using the FFA decision rule in (\ref{IFA_decision_rule}) but under the UDA. In such a set up we would anticipate that the optimal thresholds of under the FFA and UDA would be very similar at $\rho^*$.

To proceed with a comparison under the FFA and UDA we conduct Monte Carlo simulations. In these simulations, note that although $p(\bm{m}|\Hnull)$ as given by (\ref{likelihood0_one}) is used, the likelihood $p(\bm{m}|\Halt)$  given
\begin{equation}\label{numerical_Circle}
\begin{split}
p(\bm{m}|\Halt) &= \int p(\bm{m}|\bm{\theta}, \Halt)p(\bm{\theta}|\Halt) d\bm{\theta}.
\end{split}
\end{equation}
 must be determined numerically. The other simulation settings are the same as under the FFA, except the malicious users are uniformly distributed on a circle, whose radius and center are $R$ and $\bm{\theta}^c$, respectively. The measurements collected from the legitimate and malicious users are generated according to (\ref{h0_RSS}) and (\ref{h1_RSS_general}), respectively. To obtain the true numerical NMI under the UDA, we use (\ref{likelihood0_one}) and (\ref{numerical_Circle}) to calculate $p(\bm{m}|\Hnull)$ and $p(\bm{m}|\Halt)$, respectively, and utilize (\ref{Neyman}) as the decision rule.  To simulate the NMI obtained from the use of the FFA decision rule (but under the  UDA) we use (\ref{likelihood0_one}) and (\ref{likelihood1_one}) in order to implement the decision rule in (\ref{IFA_decision_rule}).

 From our results, shown in Fig.~\ref{fig:Circle}, we can see that at values of $\rho<< \rho^*$ the optimal thresholds (the values of $T_{\Lambda}$ which maximize the NMI) for the two cases are very different. However, as $\rho \rightarrow \rho^*$ the optimal threshold obtained under blindly adopting the FFA decision rule (even though the malicious user is not at infinity) is effectively the optimal value.
 Note also, the maximum values of the two NMIs and the corresponding $T_{\Lambda}^*$ are coincident when $\rho = \rho^*$, which verifies that the reference value of $\rho$ in (\ref{Ratio Bound}) is reasonable.

 \begin{figure}[t]
    \centering
        \epsfysize=7cm
        \epsffile{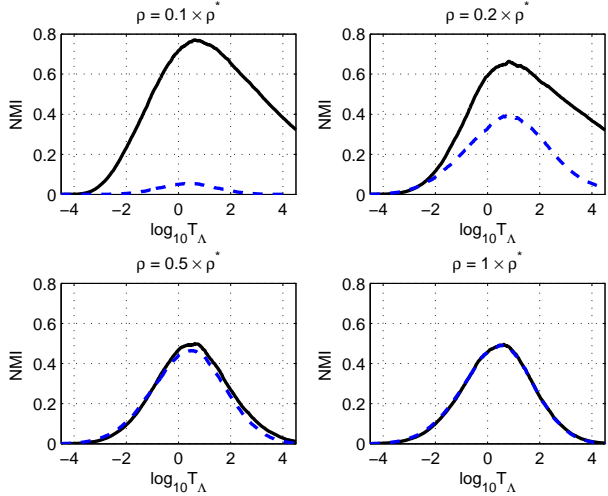}
        \caption{Normalized Mutual Information (NMI) with different values of $\rho$, where $K = 10, L = 1$, $\sigma_{dB} = 5$, and $\rho^* = 5.275$. The solid curves represent the NMI achieved under the correct decision rule (\ref{Neyman}). The dashed curves represent the NMI achieved under the FFA decision rule (\ref{IFA_decision_rule}).}
    \label{fig:Circle}
\end{figure}

\subsubsection{The DM Threat Model} \label{More General Scenario}

In this subsection we implement the optimal location verification algorithm under our adopted threat model - the DM threat model. In this model $p(\bm{\theta}|\Halt)$ is assumed to be  uniform over the annulus formed by two concentric circles, whose finite radii are $R_1$ and $R_2$ $(R_1 < R_2)$, respectively, and whose mutual center is $\bm{\theta}^c$. The use of an annulus setting allows us to cover more general settings (beyond just single highways/freeways) such as freeway intersection regions where the the freeways can have multiple directions. In any scenario (single  or intersecting roads) it will be assumed that the malicious user will not enter into any region (we assume the malicious users knows the locations of all BSs) where he is less than some distance $R_1$ from any of the VANET's infrastructure BSs  (see footnote 2).

This implies $p(\bm{\theta}|\Halt) = 1/\pi (R_2^2 - R_1^2)$, where $\bm{\theta} \in \{\bm{\theta}|R_1^2 \leq (u - u^c)^2 + (v - v^c)^2 \leq R_2^2\}$. Under this model, $p(\bm{m}|\Hnull)$ is also the same as in (\ref{likelihood0_one}), and $p(\bm{m}|\Halt)$ is as given in (\ref{numerical_Circle}) but with the modified prior distribution. Again,  no closed form solution is available for (\ref{numerical_Circle}).


 We present new Monte Carlo simulations where the settings are again the same as under the FFA, except that now  the malicious users are uniformly distributed in the annulus. Again, the measurements collected from the legitimate and malicious users are generated according to (\ref{h0_RSS}) and (\ref{h1_RSS_general}), respectively. To obtain the true numerical NMI under the DM threat model, we use (\ref{likelihood0_one}) and (\ref{numerical_Circle}) to calculate $p(\bm{m}|\Hnull)$ and $p(\bm{m}|\Halt)$, respectively, and utilize (\ref{Neyman}) as the decision rule.  To simulate the NMI obtained from the use of the FFA decision rule (but under the DM threat model) we use (\ref{likelihood0_one}) and (\ref{likelihood1_one}) in order to implement the decision rule in (\ref{IFA_decision_rule}). The results of our simulations are shown in Fig.~\ref{fig:AnnulusR2}, where $\rho = 0.2 \rho^*$, and $\rho$ is redefined as $\rho = R_1/r$.

\begin{figure}[t]
    \centering
        \epsfysize=7cm
        \epsffile{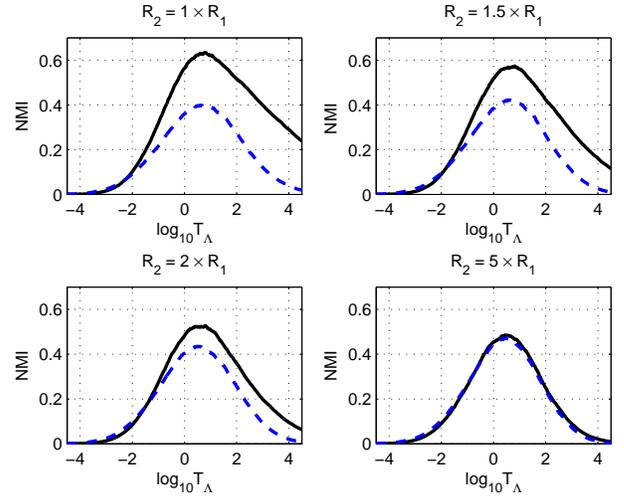}
        \caption{Numerical and IFA approximated Normalized Mutual Information (NMI) with different values of $R_2$, $\rho^* = 5.275, \rho = 0.2 \times \rho^*$, and $\sigma_{dB} = 5$. The solid curves represent the NMI achieved under the correct decision rule (\ref{Neyman}). The dashed curves represent the NMI achieved under the FFA decision rule (\ref{IFA_decision_rule}).}
    \label{fig:AnnulusR2}
\end{figure}

In the top left plot of  Fig.~\ref{fig:AnnulusR2}, we have set $R_2 = R_1$, so in this specific plot the DM threat model is equivalent to that under the UDA (the result is the same as that shown in the top right plot of Fig.~\ref{fig:Circle}). However, again we see that as $R_2$ increases the optimal threshold obtained under blindly adopting the FFA decision rule (even though the malicious users are constrained within an annulus) is effectively the optimal value. Note also, that in the DM  threat model, as $R_2$ increases this results holds for  cases when $\rho<< \rho^*$ (which was not the case under the UDA).


As a final point, we note that instead of numerically solving (\ref{numerical_Circle}), it may be useful to find an approximate closed-form solution to (\ref{numerical_Circle}) (\emph{e.g.} this would allow for an approximate closed-form for the false positive and detection rates under this threat model). We can approximate $p(\bm{m}|\Halt)$ via an application of the Laplace approximation, which can approximate integrals through a series expansion by using local information about the integrand around its maximum \cite{kass1995bayes} \cite{nevat2011cooperative}. The details are as follows. First, let us define a quantity as:
\begin{align*}
h\left(\bm{\theta}|\Halt \right) \triangleq \ln \left( p\left(\bm{m} |\bm{\theta}, \Halt \right) p\left(\bm{\theta}|\Halt \right)\right).
\end{align*}
  $h\left(\bm{\theta}|\Halt \right)$ can be expanded using a Taylor series around its maximum \textit{a-posteriori} (MAP) estimate, denoted by $\widehat{\bm{\Theta}} = \arg\underset{\bm{\theta}}{\max}\, {\{p(\bm{m} |\bm{\theta},\Halt) p(\bm{\theta}|\Halt)\}}$. This is the point where the posterior density is maximized, $i.e.$, the mode of the posterior distribution. Hence, we obtain to second order,
\begin{equation}
\label{Taylor_expansion}
\begin{split}
h\left(\bm{\theta}|\Halt \right) =~&
h\left(\widehat{\bm{\Theta}}|\Halt \right) +
\left(\bm{\theta}  - \widehat{\bm{\Theta}} \right)^T \underbrace{\frac{\partial h\left(\widehat{\bm{\Theta}} |\Halt\right)}{\partial \bm{\theta} }}_{
\left(=0\right) \text{ at MAP location}}\\
&+ \frac{1}{2} \left(\bm{\theta}  - \widehat{\bm{\Theta}} \right)^T \frac{\partial^2 h\left(\widehat{\bm{\Theta}} |\Halt\right)}{\partial^2 \bm{\theta} }
\left(\bm{\theta}  - \widehat{\bm{\Theta}} \right).
\end{split}
\end{equation}
The second term in (\ref{Taylor_expansion}) is $0$, because the first derivative is zero at the maximum of $h\left(\bm{\theta} |\Halt \right)$. Replacing $h\left(\bm{\theta}|\Halt \right)$  by the truncated second-order Taylor series yields:
\begin{align*}
h\left(\bm{\theta}|\Halt \right) \approx
h\left(\widehat{\bm{\Theta}}|\Halt \right) +
\frac{1}{2} \left(\bm{\theta}  - \widehat{\bm{\Theta}} \right)^H \mathbb{\H} \left(\bm{\theta}  - \widehat{\bm{\Theta}} \right),
\end{align*}
where $\mathbb{\H}$ is the Hessian of the $\ln$ posterior, evaluated at $\widehat{\bm{\Theta}}$:
\begin{align*}
\mathbb{\H}\triangleq  \left.\frac{\partial^2 h\left(\widehat{\bm{\Theta}}|\Halt \right)}{\partial^2 \bm{\theta} }\right|_{\bm{\theta}=\widehat{\bm{\Theta}} }=
\left. \frac{\partial^2 h\left(\bm{\theta}|\Halt \right)}{ \partial \bm{\theta}   \partial \bm{\theta}^H} \right|_{\bm{\theta}=\widehat{\bm{\Theta}} }.
\end{align*}
Using the above approximation, we have the following
\begin{align*}
\begin{split}
\ln p\left(\bm{m}|\Halt \right)
=& \ln \int p\left(\bm{m} |\bm{\theta}, \Halt\right)
						 p\left(\bm{\theta}|\Halt\right) d \bm{\theta}\\
=& \ln \int 						
						 \exp^{h\left(\bm{\theta}|\Halt \right)} d \bm{\theta}\\
\approx& \ln \int 						
						 \exp^{h\left(\widehat{\bm{\Theta}} |\Halt\right) +
\frac{1}{2} \left(\bm{\theta}  - \widehat{\bm{\Theta}} \right)^T \mathbb{\H} \left(\bm{\theta}  - \widehat{\bm{\Theta}} \right)} d \bm{\theta}\\
=& h\left(\widehat{\bm{\Theta}}|\Halt \right) +\ln \int 						
						 {\exp^{\frac{1}{2} \left(\bm{\theta}  - \widehat{\bm{\Theta}} \right)^T \mathbb{\H} \left(\bm{\theta}  - \widehat{\bm{\Theta}} \right)}}d \bm{\theta}\\
=&  h\left(\widehat{\bm{\Theta}}|\Halt \right)+\frac{1}{2} \ln \left|2 \pi \mathbb{\H}\right|\\
=&  \ln p\left(\widehat{\bm{\Theta}}|\Halt \right)
+\ln p\left(\bm{m}|\widehat{\bm{\Theta}},\Halt \right)\\
&+ \ln \left|2 \pi \mathbb{\H}^{-1}\right|^{1/2}.
\end{split}						
\end{align*}
Finally, the marginal likelihood estimate can be written as
\begin{align}
\label{laplace_est}
\widehat{p}\left(\bm{m}|\Halt\right)=
p\left(\widehat{\bm{\Theta}}|\Halt \right)
p\left(\bm{m} |\widehat{\bm{\Theta}}, \Halt \right)
\left|2 \pi \mathbb{\H}^{-1}\right|^{1/2}.
\end{align}
In (\ref{laplace_est}), $p\left(\widehat{\bm{\Theta}}|\Halt \right)$ and $\left|2 \pi \mathbb{\H}^{-1}\right|^{1/2}$ are both constant for a specific $\widehat{\bm{\Theta}}$; thus the Laplace approximated likelihood function, $\widehat{p}\left(\bm{m}|\Halt\right)$, is a $K$ dimensional normal distribution with the same variance as $p(\bm{m}|\Hnull)$ (because the variances of $p(\bm{m}|\Hnull)$ and $p(\bm{m}|\bm{\theta}, \Halt)$ are the same). Under the Laplace approximation the decision rule in (\ref{Neyman}) is approximated by
\begin{equation}\label{Laplace_decision_rule}
\widehat{\Lambda}\left(\bm{m}\right) = \frac{\widehat{p}\left(\bm{m}|\Halt\right)}{p\left(\bm{m}|\Hnull \right)}
\begin{array}{c}
\overset{\Hoalt}{\geq} \\
\underset{\Honull}{<}
\end{array}%
T_{\Lambda}.
\end{equation}

\begin{figure}[t]
    \centering
        \epsfysize=7cm
        \epsffile{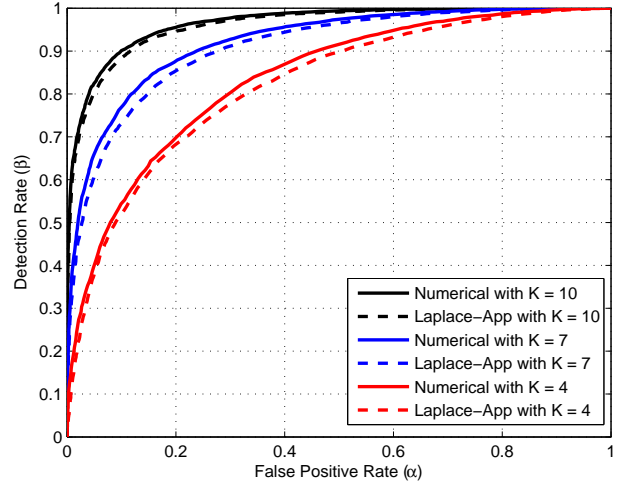}
        \caption{Numerical and Laplace approximated ROC curves with different values of $K$, where $\sigma_{dB} = 5, L = 1, R_1 = 100 m$, and $R_2 = 500 m$.}
    \label{fig:Laplace_ROC}
\end{figure}

To study the performance our Laplace approximation we calculate ROC curves for both the numerical Monte Carlo calculation of $p(\bm{m}|\Halt)$ and the Laplace approximation of $p(\bm{m}|\Halt)$. These different forms are then used in the same decision rule (\ref{Neyman}) in order to form the ROC curves. The results of these simulations are shown in Fig.~\ref{fig:Laplace_ROC}, and as we can see the approximation is a good one for the parameters used. We have further investigated the accuracy of the approximation over a range of other parameters, finding  similar results to those shown in Fig.~\ref{fig:Laplace_ROC}.

\subsection{Discussion}

In the preceding sections we have looked at a general attack scenario under specific threat models. The attack scenario we have focussed on is that of a non-colluding adversary who is attempting to spoof he is within the perimeter of some wireless network, when in reality he is some distance beyond the network boundary. A non-colluding adversary who is within the network region will in general be easily identified due to his inability to set different received signal strengths  at different BSs. An attack from outside the network region is the perhaps most realistic and likely scenario one can imagine for the emerging  ITS scenario. For example, a single adversary who is some distance from the highway (so as not to be easily  identified or caught) is attempting to disrupt proper functioning of the ITS on the highway.

What we have shown through our investigations of specific threat models under our general attack scenario, is that an optimal LVS can be developed for each threat model, but in a non-straightforward manner in most cases - \emph{i.e.} no closed-form solutions for the detection and false positive rates are available. Without such closed-form solutions for these rates, one must resort to complex and time consuming Monte Carlo simulations in order to determine the optimal threshold.  However, from considerations of the  FFA, and how other threat models can be approximated under the FFA, we have shown how a straightforward LVS algorithm can be deployed which is effectively optimal for most circumstances. More specifically, using analytical solutions for the detection and false positive rates in the FFA setting, which are then used in easily determining the optimal threshold value, a straightforward LVS is developed whose performance is near-optimal when the adversary is close to the network, and optimal as the adversary  moves to a large distance from network.

However, of course more sophisticated attacks than those highlighted above are possible. The most obvious of these is that of \emph{colluding} adversaries who can communicate and cooperate with each other so as to form collective attacks on the LVS. An example of such an attack would be colluding adversaries who set different received signal strengths at different BSs. On the defensive side, the network could also deploy beam-forming techniques to help the LVS thwart these types of attacks. The LVS could also deploy tracking algorithms and physical layer security techniques to assist in its defense. These more sophisticated forms of attack and their corresponding defensive strategies are out of scope of the current work, but do form part of our ongoing research efforts in this area. However, we should be clear that ultimately \emph{any} defensive strategy for an LVS is ultimately doomed if the colluding adversary is afforded unlimited resources and the communications network is purely classical in nature.\footnote {Note that location verification in the context of quantum communications
systems has previously been considered e.g. \cite{kent2006tagging}, \cite{malaney2010location}, \cite{malaney2010quantum}, and it has been
argued that such systems are able to securely verify a location under all
known threat models \cite{malaney2011location} - although see \cite{buhrman2011position} who argue otherwise. It is
undisputed that classical communications alone cannot achieve secure location
verification under all known threat models.}

In this work, we assume error free location estimation for the legitimate users. In fact, if the localization error is small relative to the scale of the network boundary, the effects of this error can be ignored. To verify this, we have carried out additional Monte Carlo simulations identical to those producing Fig.~5 except the localization error for  legitimate users is assumed to follow a bivariate normal distribution. We find the results are negligibly different from those shown in Fig.~5 even if the variance of the localization error in each coordinate is 100 $m^2$. It is perhaps worth noting that the prior distributions of the localization error and a malicious user's location are different, and we could not distinguish between them if there is an overlap between the two distributions. Inaccurate knowledge of the system and channel parameters will reduce the  LVS performance. We have quantified this for our specific LVS by carrying out additional simulations in which knowledge of the input LVS parameters are modified by up to $50\%$ from the true underlying parameters. We find that such errors induce an 40$\%$ error in the value of the optimal threshold. However, we do point out that this is a worse case scenario as we have assumed in our simulations that the attacker retains perfect knowledge of the parameters (which in reality will be untrue) so as to perform the optimal attack (optimal power boost). We should also note that any spatial correlation of shadowing beyond that accounted for by the shadowing standard deviation has not been included in our simulations. Any such correlation would add additional uncertainty into any LVS system unless it had been pre-measured and included as part of the channel model.

\subsection{Results in Relation to Other Works}

Location verification has been an active research area, and many verification  algorithms have been proposed for  VANETs \emph{e.g.} \cite{sastry2003secure,leinmuller2006position, xiao2006detection, song2008secure, yan2009providing, ren2009location, yan2010cross, abumansoor2012secure}, wireless sensor networks \emph{e.g.}  \cite{vora2006secure, wei2012lightweight}, and  generic wireless networks  \emph{e.g.} \cite{malaney2006secure, malaney2007securing, malaney2007wireless, capkun2008secure, zhang2008evaluation, chen2010detecting, liu2010node, chiang2012secure, yang2012securing}.

Perhaps the most closely related works to ours are those which propose optimizing the system's threshold by minimizing  the probability of misclassification (\emph{e.g.} \cite{chen2010detecting}), which is defined as $P_e = P_0 \alpha + (1-P_0)(1- \beta)$. Of course, a direct comparison between such systems and ours is not entirely meaningful, due to the different optimization metrics being used. Further to this, it is important to note the complex interplay between the entropy of a random variable, $H(X)$, and the probability of misclassification.  Although it may seem at first counter-intuitive, the fact is that there is not a one-to-one relationship between $H(X)$ and $P_e$. That is, two random variables with the same entropy can have different $P_e$ \cite{feder1992universal}. This same issue extends to NMI and $P_e$, and in the context of our LVS, it is important to recognize this fact. As such, if optimization of $P_e$ is the system  objective, then use of a Bayesian hypothesis,  where the costs of all types of classifications are equal, will suffice.\footnote{In a more general Bayesian framework, the average Bayesian cost is defined as $C = P_0 \alpha C_0 + (1 - P_0) (1- \beta) C_1$, where $C_0$ is the pre-assigned cost of rejecting a legitimate user, and $C_1$ is the pre-assigned cost of accepting a malicious user \cite{barkat2005signal}. If $C_0$ and $C_1$ are known or can be set, the Bayesian framework is optimal for an LVS in the sense that it minimizes the average Bayesian cost ($P_e$ is the special case of $C$ with $C_0=C_1=1$). } But again we must stress that in the context of real-world LVS deployments this represents a strong \emph{subjective} decision on the cost of misclassifications. Given the complexity, and the many different roles of location information within the ITS scenario (crash avoidance, vehicle-congestion avoidance, vehicle-to-vehicle communication protocols \emph{etc}.), proper determination of misclassification costs will be, at best, extremely complex in nature. It is for this reason we have approached optimization of our LVS from an \emph{objective} information-theoretic viewpoint. Our guiding light has been the well-known Infomax principle \cite{linsker1988self}, which states an optimal system must transfer as much information as possible from its input to its output - \emph{i.e.} maximizes the mutual information between its  inputs and outputs.

\begin{figure}[t]
    \centering
        \epsfysize=7cm
        \epsffile{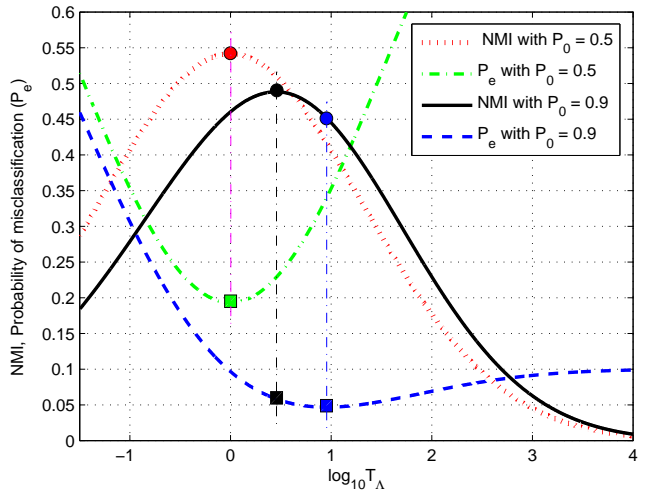}
        \caption{Normalized Mutual Information (NMI) and probability of misclassification ($P_e$) with different values of $T_{\Lambda}$. (Note here the system and network parameters are same as those utilized in Fig.~\ref{fig:IFA}.)}
    \label{fig:comparison}
\end{figure}

 Notwithstanding the above discussion, we compare the optimal thresholds for NMI and $P_e$, as shown in Fig.~\ref{fig:comparison}. From this figure, we can see that for $P_0=0.5$ the optimal threshold ($T_{\Lambda}=1$) is the same for both algorithms. However, the optimal thresholds for the two algorithms are different when  $P_0 = 0.9$. Further, we see that in the $P_0 = 0.9$ case, the change in $P_e$, if the optimal NMI threshold is used instead of the optimal $P_e$ threshold, is significantly less than the change in NMI if the optimal $P_e$ threshold is used instead of the optimal NMI threshold.

 We expand upon this last point in Fig.~\ref{fig:P0} where the optimal thresholds for each system are plotted as  functions of $P_1=1-P_0$.
 This figure  outlines another pragmatic advantage of the NMI  approach.
  In  reality, the base rate of intrusions ($P_1=1-P_0$)  is an unknown parameter \emph{for all} LVS systems. As such, we see from Fig.~\ref{fig:P0} that the use of NMI results in a  more robust system. As the true value of $P_1$ approaches small values (in any real situation it will be small)  the NMI threshold is insensitive to the assumed $P_1$. This means that when using NMI as the optimization metric, any  mismatch between true $P_1$ and assumed $P_1$ has little impact on system performance. In the Bayesian framework, however, the optimal threshold  for minimizing $P_e$ remains very sensitive (linear) to the assumed value of $P_1$. In this latter case, a mismatch between true $P_1$ and assumed $P_1$ results in very poor system performance.

\begin{figure}[t]
    \centering
        \epsfysize=7cm
        \epsffile{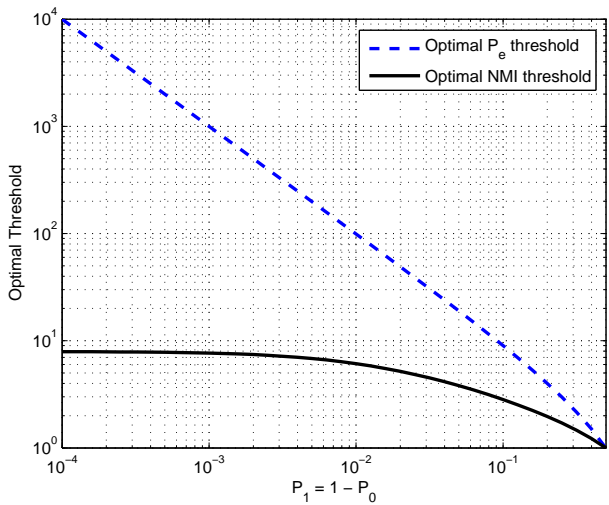}
        \caption{Optimal NMI threshold and optimal $P_e$ threshold as functions of $P_1 = 1 - P_0$. (Note here the system and network parameters are same as those utilized in Fig.~\ref{fig:IFA}.)}
    \label{fig:P0}
\end{figure}

\section{Conclusions} \label{Conclusions}
In this paper, we developed  an information-theoretic framework for an LVS, utilizing as the objective optimization criterion  the mutual information between input and output data of the LVS. We investigated our new  optimal LVS under a  realistic threat model, showing how in a  straightforward implementation of an LVS, information-theoretic optimality is approached as the non-colluding adversary moves further from the network region it is claiming to be within. This straightforward implementation makes our new algorithm an ideal candidate for the LVS that will be needed in emerging network-based and safety-enhanced transportation systems, such as ITS.

\section{Acknowledgment}

This work has been supported by the University of New
South Wales, and the Australian Research Council, grant DP120102607. Gareth W. Peters is supported in this research by a Royal Society International Exchanges Scheme  IE121426.

\end{document}